\documentclass[runningheads]{llncs}
\usepackage[T1]{fontenc}
\usepackage{graphicx}
\usepackage{geometry}
\usepackage{booktabs}
\usepackage{algorithm}
\usepackage{algorithmic}
\usepackage{tcolorbox}

\usepackage[switch]{lineno}
\usepackage{tikz}
\usepackage{adjustbox}
\usepackage{xcolor}
\usetikzlibrary{automata}
\usetikzlibrary{calc,arrows.meta,positioning}

\usepackage{forest}
\usetikzlibrary{arrows.meta}
\usetikzlibrary{matrix}
\usepackage{array}
\usetikzlibrary{tikzmark,arrows,calc}
\usepackage{booktabs}
\definecolor{hotmagenta}{rgb}{1.0, 0.11, 0.81}

\usepackage{amsfonts,mathtools}
\usepackage{comment,cprotect,xcolor,xspace}
\usepackage{multirow}
\usepackage{cleveref} 
\usepackage[sort,numbers]{natbib}
\usepackage{authblk}

\makeatletter 
\renewcommand\@biblabel[1]{#1} 
\makeatother

\newcommand{\profileset}{\mathbf{{A}}} 
\newcommand{\profile}{\mathcal{A}} 

\newcommand{\plausible}{\emph{plausible}\xspace}

\newcommand{\SWMProb}{\textsc{SWM-Prob}${(W,p)}$\xspace}
\newcommand{\isPossSWM}{\textsc{IsPossSWM}\xspace}
\newcommand{\isNecSWM}{\textsc{SWM-Prob}${(W,1)}$\xspace}

\newcommand{\existsNecSWM}{\textsc{ExistsSWM-Prob}$(1)$\xspace}
\newcommand{\SWDist}{\textsc{SW-Dist}\xspace}
\newcommand{\maxSWM}{\textsc{ExistsSWM-Prob}$(p)$\xspace}
\newcommand{\ExistsSWMProb}{\textsc{ExistsSWM-Prob}\xspace}

\newcommand{\maxExpSW}{\textsc{MaxExpSW}\xspace}
\newcommand{\Prob}{\xspace\mathrm{Pr}}

\newcommand{\SWM}{\operatorname{SWM}}
\newcommand{\SW}{\operatorname{SW}}

\newcommand{\E}{\mathbb{E}}
\newcommand{\AS}{\operatorname{AS}}
\newcommand{\DP}{\operatorname{dp}}

\begin{document}
	\title{Social Welfare Maximization\\ in Approval-Based Committee Voting under Uncertainty}
	%
	%
	\author{Haris Aziz\inst{1} \and
		Yuhang Guo\inst{1} \and
		Venkateswara Rao Kagita\inst{2} \and Baharak Rastegari\inst{3} \and Mashbat Suzuki\inst{1}}
	\authorrunning{H. Aziz et al.}
	%
	\institute{University of New South Wales, Sydney, Australia \\ \email{\{haris.aziz,yuhang.guo2,mashbat.suzuki\}@unsw.edu.au}\and
		National Institute of Technology, Warangal, India \\
		\email{venkat.kagita@nitw.ac.in} \and
		University of Southampton, Southampton, UK \\
		\email{b.rastegari@soton.ac.uk}}
	\maketitle              
	\begin{abstract}
Approval voting is widely used for making multi-winner voting decisions. The canonical rule (also called Approval Voting) used in the setting aims to maximize social welfare by selecting candidates with the highest number of approvals. We revisit approval-based multi-winner voting in scenarios where the information regarding the voters' preferences is uncertain. We present several algorithmic results for problems related to social welfare maximization under uncertainty, including computing the social welfare probability distribution of a given outcome, computing the probability that a given outcome is social welfare maximizing, computing an outcome that is social welfare maximizing with the highest probability, and understanding how robust an outcome is with respect to social welfare maximization.
		\keywords{Committee Voting \and Uncertain Preference \and Social Welfare Maximization}
	\end{abstract}
	%
	%
	%
	

\section{Introduction}

Approval voting is one of the simplest and most widely used methods of making selection decisions. 
Due to its fundamental nature, it has found applications in recommender system \citep{CPG+19a,GaFa22a,SLB+17a}, blockchains \citep{BCC+20a,BBC+24a}, and Q \& A platforms \citep{IsBr24a}.
In approval voting, voters are asked to identify the candidates they approve of from a given set. 
The candidates with the highest number of approvals are then selected. Therefore, ``approval voting'' not only specifies the format of the ballots but also commonly points to the method for selecting the candidates~\citep{KILG10a,KiMa12a}. Many organizations and societies use approval voting to select committees. For example, the Institute of Electrical and Electronics Engineers (IEEE), one of the largest scientific and technical organizations, has been using approval voting for selection decisions. 

If approvals of voters are interpreted as voters' binary preferences over candidates, then the outcome of the approval voting method has a clear \textit{utilitarian social welfare} perspective: identify the set of candidates that provide the highest social welfare to the voters. We explore this utilitarian social welfare perspective when there is uncertain information regarding voters' preferences. Uncertain approval preferences are useful when the central planner only has imprecise information about the voters' preferences. This estimated information could be based on historical preferences, past selections, or online clicks or views. For example, if an agent $i$ has selected a certain candidate $c$ 70\% of the time in previous situations, one could use this information to assume that the approval probability of agent $i$ for candidate $c$ is 0.7. The uncertain information could also be based on situations where each agent represents a group of people who may not have identical approval preferences. For example, if 60\% of the group represented by agent $i$ approved a certain candidate $c$, one could assume that the approval probability of agent $i$ for candidate $c$ is 0.6. Uncertainty becomes a prevalent factor also when employing methods such as machine learning or recommendation techniques to forecast the unobserved (dis)approvals of voters for candidates. (The motivating examples are by \citet{AKR+24a}.)

We consider four different types of uncertain approval preferences that have been studied in recent work (see, e.g., \citep{AKR+24a}). 
Under the \emph{Candidate Probability model}, there is a probability for a given voter approving a given candidate. This model captures the examples in the previous paragraph. The Three Valued Approval (3VA) model ~\citep{IIBK22a} is a restricted version of the candidate probability model and captures a natural form of uncertainty where a voter has no or too little information on some candidates and assigns approval probability of  $0.5$ to such candidates.
Under the \emph{Lottery model}, each voter has an independent probability distribution over approval sets. Under the \emph{Joint Probability model}, there is a probability distribution of approval profiles. 
These last two models allow us to capture richer forms of uncertainty where there could be dependencies between candidates and between voters' approval sets. The Joint Probability model, in particular, may not seem practical; we include it for completeness.

In classical approval-based committee voting, the social welfare of any given committee can be computed exactly, making it straightforward to evaluate how ``good" the committee is. However, in scenarios involving uncertain preferences, 
the social welfare of the given committee becomes a random variable and a natural, fundamental question arises: What is the probability that a given committee achieves a sufficiently high level of social welfare? More generally, what is the distribution of the social welfare associated with the committee? Understanding the distribution is essential, as it captures the likelihood of the given committee achieving desirable social welfare.
With social welfare as the objective, the optimal committee is naturally the one that maximizes social welfare. Under deterministic preferences, this problem is well-understood. For any given plausible approval profile, the Approval Voting (AV) rule efficiently identifies the committee that maximizes social welfare in polynomial time. However, in scenarios under uncertain preference models, for example, the Lottery model or Candidate Probability model, a desirable committee in this context is not just one that performs well in a single realization, but one that maximizes social welfare in the largest fraction of all possible realizations, that is, the committee with the highest probability of being social welfare maximizing. When facing uncertainty, another intriguing and practically relevant question concerns the robustness of a committee. For example, under uncertain preferences, does there exist a committee that achieves at least half of the optimal social welfare with high probability? 
This shift from deterministic to probabilistic evaluation introduces significant computational challenges. Since the number of plausible approval profiles grows exponentially with the number of agents and candidates, problems related to social welfare maximizing committees are computationally challenging. 
This leads to compelling questions:
\begin{quote}
	\textit{
		How can we compute the distribution of a given committee's social welfare?
		Can we compute the probability of a given committee being social welfare maximizing in polynomial time?
		Can we efficiently identify the committee that maximizes social welfare (or well-approximates optimal social welfare) with the highest probability?}
\end{quote}

\subsection{Our Results}

Our first contribution is polynomial-time algorithms for computing the distribution of a given committee's social welfare (\SWDist). We design algorithms both for  Lottery and Candidate probability models.
These results allow us to check whether a committee achieves high level of social welfare with high probability.

We next turn to the problem of computing the probability that a given committee is social welfare maximizing (\SWMProb).  
We show that under the Candidate Probability model, the problem is solvable in polynomial time.
In contrast, we prove that the problem is NP-complete under the Lottery model. 
We also present a positive result for the decision variant that asks whether a given committee is social welfare maximizing with probability one, showing that it is solvable in polynomial time.

Next, we explore the problem whether there exists a committee whose probability of being social welfare maximizing exceeds a given threshold $p$ (\maxSWM). 
We show that this problem is NP-hard under the Lottery model. Nevertheless, we provide a polynomial-time result for the special case where $p = 1$.
For the Candidate Probability model, the problem becomes intriguing as even for the  more restrictive 3-Valued Approval (3VA) setting, there is no known characterization of committees that maximize the probability of being social welfare maximizing. While the complexity of the problem remains open, we present positive results under certain constraints, for instance, when the committee size $k$ is constant, or when $p = 1$.

We finally consider the problem of robust welfare maximization: does there exist a committee which guarantees a fraction of the optimal social welfare of the realized profile with high probability? We provide a positive result for the 3VA model and an impossibility result for the Candidate Probability model. Our key results are summarized in Table~\ref{table:summary:uncertainABC-SW}. Missing proofs are relegated to the appendix.  
\begin{table}[ht]
	\centering
		\begin{tabular}{ccc}
			\toprule
			\textbf{Problems} &\textbf{Lottery Model }&\textbf{Candidate Probability Model}  \\
			\midrule
			\SWDist &  in P (Thm.~\ref{thm:SWDist-lot}) & in P (Thm.~\ref{thm:SWDist-CP})  \\
			\midrule 
			\multirow{2}{*}{\SWMProb} & NP-h (Thm.~\ref{SWM-Prob-Sharp-P-c}) &\multirow{2}{*}{in P (Thm.~\ref{thm:MaxSW-Prob-CP-3VA})} \\   
			& (in P (Thm.~\ref{thm:IsNecSWM-Lot})  when $p=1$)  & \\
			\midrule                        
			\multirow{2}{*}{\maxSWM} &  NP-h (Thm.~\ref{th:MaxSWM-Lottery})   & \multirow{2}{*}{$?$}    \\                
			& (in P (Thm.~\ref{thm:ExistsNecSWM-Lot})  when $p=1$) &   \\ 
			\bottomrule
	\end{tabular}
	\caption{Summary of results.}
	\label{table:summary:uncertainABC-SW}
\end{table}

\subsection{Related Work}
Approval-Based Committee (ABC) voting has received considerable attention in recent years (see, e.g., \citep{KiMa12a,ABC+17a,BFJ+17a,AEH+18a,LaSk23a}), primarily focusing on selecting “proportional” committees. 
Concurrently, there has been growing interest in preference aggregation under \emph{uncertainty}. \citet{KoLa05a} study winner determination problems with incomplete preferences for ranking-based single-winner voting rules.  \citet{HAK+12a} explore the probability of a particular candidate winning an election under uncertain preferences for various voting rules, such as the Plurality rule and the Borda rule. \citet{BORO16a} provide a survey on uncertainty and communication in voting. \citet{DHL+22a} investigate the dynamic selection of candidates, where uncertainty is related to the order in which candidates appear. \citet{HKPTW23} devise different query algorithms for scenarios where voters’ ballots are partial and incomplete over all candidates, while \citet{BDI+23a} propose an ABC voting model with possibly unavailable candidates and examine voting rules which admit “safe” query policies to check candidates’ availability. 

Highly relevant to our paper, \citet{BLMR13} examine single and multi-winner voting under uncertain approvals, and \citet{TKO21a} address the problem of checking whether an incomplete approval profile admits a completion within a certain restricted domain of approval preferences. \citet{IIBK22a} study several computational problems, including checking whether a given committee is possibly or necessarily \textit{Justified Representation} (JR) or whether there is a possible outcome of various rules, including Approval Voting. Of the uncertain preference models that we consider, \citet{IIBK22a} have explored the 3VA model. Most of the computational problems that we consider are not studied by \citet{IIBK22a}. They presented a polynomial-time algorithm to check whether given a committee is social welfare maximizing under some realization of the uncertain preferences. This result is implied by our more general results.
Recently, \citet{AKR+24a} consider several probabilistic preference models and problems, such as maximizing the probability of satisfying JR. In our paper, we consider the fundamental objective of maximizing social welfare under the same uncertain preference models.

Uncertain preferences have also been studied in other domains, such as matching and fair allocations. \citet{ABH+19a} investigate the computational complexity of Pareto optimal allocation under uncertain preferences. \citet{RSS24} explore the envy-free allocation for the house allocation market. \citet{ABG+20a} address the problem of computing stable matchings with uncertain preferences. Additionally, \citet{BDE+24a} examine stable matchings under one-sided uncertainty, focusing on computational issues within three different competitive query models.   


\section{Preliminaries}
For any~$t \in \mathbb{N}$, let $[t] \coloneqq \{1, 2, \dots, t\}$.
An \emph{instance} of the (deterministic) approval-based committee (ABC) voting is represented as a tuple $(V, C, \profile, k)$, where:
\begin{itemize}
	\item $V = [n]$ and $C = [m]$ are the sets of \emph{voters} and \emph{candidates}, respectively.
	
	\item Each voter approves a set of candidates $A_i \subseteq C$. Let $\profile = (A_1, A_2, \dots, A_n)$ denote voters' approval profile.
	The set of all possible approval profiles is denoted by~$\profileset$.
	
	\item $k$ is a positive integer that represents the committee size.
\end{itemize}
A \emph{feasible winning committee}~$W \subseteq C$ is of size~$k$. In this paper, we consider only feasible winning committees unless otherwise specified. Given a committee $W$, the welfare of each voter $i$ is the number of candidates in $W$ of whom $i$ approves. The \emph{Social Welfare} ($\SW$) of $W$ given the approval profile $\profile$ is defined as the sum of the welfare of the voters, $\SW(W, \profile) = \sum_{i\in V} |W\cap A_i|$.
Given an approval profile $\profile = (A_1, A_2, \dots, A_n)$, for each candidate $c\in C$, we denote the approval score of candidate $c$ by $\AS(c, \profile) = |\{i \in V :   c \in A_i\}|$.
A committee $W$ is \emph{Social Welfare Maximizing} ($\SWM$) under approval profile $\profile$~if $W$ generates the maximum social welfare among all committees of size $|W|=k$. Given an approval profile $\profile$, we can compute an $\SWM$ committee in polynomial time by computing each candidate's approval score and selecting the $k$ candidates with the highest approval scores in a greedy manner. Consequently, given a deterministic approval profile $\profile$ and a committee $W$, we can decide in polynomial time whether or not $W$ is $\SWM$.

\subsection{Uncertain Preference Models}\label{sec:uncertain-models}
\label{ssec:uncertain-preference-models}
Our main focus is on ABC voting under \emph{uncertain} approval ballots, where there is an underlying probability distribution over approval profiles. We adopt the following uncertainty models considered by \citet{AKR+24a}.
\begin{enumerate}
	\item \textbf{Joint Probability model:}
	A probability distribution~$\Delta(\profileset) \coloneqq \{(\lambda_r, \profile_r)\}_{r \in [s]}$ is given over~$s$ possible approval profiles with $\sum_{r \in [s]} \lambda_r = 1$, where for each~$r \in [s]$, the approval profile~$\profile_r$ is associated with a positive probability~$\lambda_r>0$. {We write $\Delta(\profile_r) = \lambda_r$.}
	
	\item \textbf{Lottery model:} For each voter $i\in V$, we are given a probability distribution~$\Delta_i \coloneqq \{(\lambda_r, S_r)\}_{r \in [s_i]}$ over~$s_i$ approval sets with $\sum_{r \in [s_i]} \lambda_r = 1$.  For each~$r \in [s_i]$, voter $i$ approves (exactly) candidate set~$S_r\subseteq C$  with probability~$\lambda_r>0$. {We write $\Delta_i(S_r) = \lambda_r$.} 
	We assume that the probability distributions of voters are independent.
	
	\item \textbf{Candidate Probability model:} Each voter $i$ approves each candidate~$c$ 
	\textit{independently} with probability~$p_{i, c}$, i.e., for each~$i \in V$ and each~$c \in C$, $p_{i, c} \in [0, 1]$. The \textbf{Three Valued Approval (3VA) model}, is a special case where each agent specifies a subset of candidates that are approved and a subset of candidates that are disapproved, and the remaining candidates could be approved or disapproved independently with equal probability. Specifically, $\forall~ i\in V, c\in C$, $p_{i,c}\in\{0, \frac{1}{2}, 1\}$, wherein $0$ denotes disapproval, $1$ indicates approval, and $\frac{1}{2}$ represents unknown.
\end{enumerate}

We use the notation $\Delta := (\Delta_1, \Delta_2, \dots, \Delta_n)$ to represent the input of the uncertain approval model, which may follow either the Lottery model or the Candidate Probability model, depending on the context of the problem or algorithm under consideration. Under the Lottery model, each $\Delta_i \coloneqq {(\lambda_r, S_r)}_{r \in [s_i]}$ specifies a probability distribution over $s_i$ approval sets for agent $i$. For the Candidate Probability model, each element $\Delta_i \coloneqq (p_{i,1}, p_{i,2},\dots, p_{i,m})$ represents agent $i$’s independent approval probabilities over the $m$ candidates.

The Joint Probability and Lottery models have been studied in other contexts including two-sided stable matching problems and assignment problems~\citep{ABH+19a,ABG+20a}. The 3VA model has been studied in ABC voting~\citep{IIBK22a}. 
We refer to an approval profile that occurs with positive probability, under any of the uncertainty models, as a \textbf{\plausible} profile. The following facts about these uncertainty models were recently pointed out by \citet{AKR+24a}.

\begin{proposition}[\citet{AKR+24a}]\label{unique_joint_prob_for_lottery}
	There is a unique Joint Probability model representation for preferences given under the Lottery model; There is a unique Lottery model representation for preferences given in Candidate-Probability model.
\end{proposition}

\subsection{Computational Problems}\label{sec:problems}
We begin by considering the most natural \maxExpSW problem, which computes a committee that maximizes the expected social welfare ($\E[\SW(W)]$). Formally, $\E[\SW(W)]$ of a committee $W$ is defined as 
$\E[\SW(W)]= \sum_{\profile \in \profileset} \Delta(\profile) \cdot \SW(W, \profile)= \sum_{\profile \in \profileset} \Delta(\profile) \cdot \sum_{i\in V}|W\cap A_i|$. We show that it is polynomial-time solvable\footnote{Here and throughout, “polynomial-time solvable” means polynomial in the size of the input.} under any of the four studied uncertain models and defer the proofs to the appendix. After this,
we study a fundamental computational problem: given a committee $W$, how can we compute the distribution of its social welfare under an uncertain approval model? We term this problem \SWDist. Notice that for any fixed committee $W$, its social welfare is a discrete random variable that takes integer values in the range $[0,kn]$. Formally, the \SWDist problem entails computing the probability $\Prob[\SW(W) = \tau]$ for each $\tau \in [0,kn]$.

We next consider the problem of computing the exact probability that a given committee $W$ is social welfare maximizing ($\SWM$) under uncertain preferences. We formalize this as a decision problem termed \SWMProb, which is defined as follows: given a committee $W$ and a threshold $p\in [0,1]$, determine whether $W$ has at least probability $p$ of being social welfare maximizing. 
\begin{tcolorbox}[title=\SWMProb,colbacktitle=gray!20!white,coltitle=black]
	\textbf{Input}: Voters $V$, Uncertain Approval Model $\Delta$, Committee $W$, threshold $p$; \\
	\textbf{Question}: Decide whether $\Prob[ W \text{ is } \SWM  ] \geq p$. 
\end{tcolorbox}

Beyond computing the probability for a fixed committee being social welfare maximizing, a more challenging problem we investigate in this paper is how to identify a committee that has the highest probability of being social welfare maximizing. We formally formulate this as a decision problem, denoted by \maxSWM: given an uncertain approval model and a threshold $p\in [0, 1]$, decide whether there exists a committee $W$ such that the probability of $W$ being social welfare maximizing is at least $p$.
\begin{tcolorbox}[title=\maxSWM,colbacktitle=gray!20!white,coltitle=black]
	\textbf{Input}: Voters $V$, Uncertain Approval Model $\Delta$, threshold $p$; \\
	\textbf{Question}: Decide whether there exists a committee $W$ such that $\Prob[ W \text{ is } \SWM  ] \geq p$. 
\end{tcolorbox}

The rest of the paper is organized as follows. In \Cref{sec:dist} we study the \SWDist problem under the Lottery and Candidate Probability models and propose two distinct dynamic programming algorithms. With this fundamental tool in hand, we examine the \SWMProb problem in \Cref{sec:probswm} and provide hardness results under the Lottery model and polynomial-time results under the Candidate Probability model. In \Cref{sec:maxswm}, we study the \maxSWM problem.
Finally, we define robust committees 
and discuss the existence and computation of a robust committee under the Candidate Probability model in \Cref{sec:robust}.
	

\section{Computation of Social Welfare Distribution}\label{sec:dist}
We begin by examining the fundamental \SWDist problem:
Given a committee $W$, what is the distribution of its social welfare under uncertain preference models? As previously discussed, the social welfare $\SW(W)$ is a discrete random variable that takes integer values in the range $[0, kn]$. Hence, we can focus on the key computational problem: given a committee $W$ and an integer $\tau \in [0, kn]$, compute the probability $\Prob[\SW(W) = \tau]$. By solving this for all values of $\tau$ in the range $[0, kn]$, we can fully characterize the distribution of $\SW(W)$.
Under the Lottery and Candidate Probability models, we propose two distinct dynamic programming algorithms, each solving the \SWDist problem in polynomial time.

\begin{theorem}\label{thm:SWDist-lot} 
	Under the Lottery model, \SWDist is solvable in polynomial time.
\end{theorem}

\begin{proof}
	Given a committee $W$ and a deterministic approval profile $\profile$, $\SW(W)$ is the sum of $|W \cap A_i|$ for each voter $i$. Under the Lottery model, each voter's approval set is random and hence the value of $|W \cap A_i|$ is a random variable, denoted by $f_i(W)$. Then, we have $\Prob\left[\SW(W) = \tau \right] = \Prob [\sum_{i\in [n]} f_i(W) = \tau]$.
	Conditioning on $f_n(W)$, it can be reformulated as $\sum_{r=0}^{\tau}\Prob\Big[\sum_{i\in [n-1]}f_i(W)=r\mid f_n(W)=(\tau - r)\Big] \cdot \Prob [f_n(W)=(\tau - r)].$
	Note that for any two voters $i$ and $j$, $f_i(W)$ is independent of $f_j(W)$ as each voter's approval set is sampled independently. Therefore  
	\begin{equation}\label{dp_eq_lottery}
		\begin{aligned}
			\Prob\left[\SW(W) = \tau\right] &= \Prob \left[\sum_{i\in [n]} f_i(W) = \tau\right] \\
			& =\sum_{r=0}^{\tau}\Prob\left[\sum_{i\in [n-1]}f_i(W)=r\right] \cdot \Prob[f_n(W)=(\tau-r)].
		\end{aligned}
	\end{equation}
	We design \Cref{dp_algo_sw_w_equal_t_lottery_model} by leveraging \Cref{dp_eq_lottery}.
	\begin{algorithm}[!htbp]
		\caption{\SWDist algorithm for the Lottery model}
		\label{dp_algo_sw_w_equal_t_lottery_model}
		\begin{algorithmic}[1]
			\REQUIRE $W$, $k$, $\tau$, $\Delta$  \\
			\ENSURE $\DP[n][\tau]$
			\STATE Initialize an $n\cdot (k+1)$ matrix $f$ with elements $0$;
			\FOR{each voter $i\in V$}
			\FOR{integer $j=(0,\ldots,k)$}
			\STATE $f[i][j] \leftarrow \sum_{r\in [s_i]}\lambda_r \cdot \mathbb{I}[| W\cap S_r |=j]$;
			\ENDFOR
			\ENDFOR
			\STATE Initialize an $n \cdot (nk+1)$ matrix $\DP$ with elements $0$;
			\STATE \textbf{for} $t=(0,\ldots, k)$, $\DP[1][t] \leftarrow f[1][t]$;
			\FOR{$i \leftarrow (2,\ldots, n)$}
			\FOR{$t \leftarrow (0, \ldots, nk)$}
			\STATE $\DP[i][t] = \sum_{r=0}^t \DP[i - 1][r] \cdot f[i][t-r]$;
			\ENDFOR 
			\ENDFOR 
			\STATE Return $\DP[n][\tau]$.
		\end{algorithmic}
	\end{algorithm}	
	In \Cref{dp_algo_sw_w_equal_t_lottery_model}, we denote $f[i][j]$ as the probability of the event $f_i(W)=j$ (line 2-6). Specifically, $f[i][j]$ represents the sum of the realization probabilities $\lambda_r$ of the deterministic approval sets $S_r$ where the size of the intersection with $W$ is exactly $j$.
	After pre-processing, we initialize the dynamic programming matrix $\DP$, where $\DP[i][j]$ is the probability of $\sum_{\ell \in [i]} f_\ell(W)=j$. Note that, computing $\Prob[\SW(W)=\tau]$  is equivalent to determining the value of $\DP[n][\tau]$. The recursive relation in \Cref{dp_eq_lottery} corresponds to $\DP[n][\tau] = \sum_{r=0}^\tau \DP[n-1][r] \cdot f[n][\tau-r]$. Starting from $\DP[1][0]$, we compute each value in the $\DP$ matrix recursively (line 8-13). Since the algorithm runs in $O(nk\cdot \max(k,\max_{i\in N}|s_i|))$ time, it implies that \SWDist under the Lottery model is solvable in polynomial time.
\end{proof}

For the Candidate Probability model, we prove that \SWDist is also polynomial-time solvable using a different dynamic programming approach.
\begin{theorem}\label{thm:SWDist-CP}
	Under the Candidate Probability model, \SWDist is solvable in polynomial time.
\end{theorem}

\begin{proof}
	Under the candidate probability model, given a committee $W$, for each voter $i$ and each candidate $c\in W$, $p_{i,c}$ falls into three different cases:
	\begin{itemize}
		\item $p_{i,c}=0$, voter $i$ certainly disapproves candidate $c$;
		\item $p_{i,c}=1$, voter $i$ certainly approves candidate $c$;
		\item $p_{i,c}\in (0,1)$, voter $i$ approves candidate $c$ with a uncertain probability $p_{i,c}$. 
	\end{itemize}
	
	We first denote $n^1=|\{(i,c): i\in V,c\in W, p_{i,c}=1\}|$ as the number of certain approvals while $n^u=|\{(i,c): i\in V,c\in W, p_{i,c}\in (0,1)\}|$ as the number of uncertain approvals. Because of the existence of the uncertain approvals, the social welfare $\SW(W)$ of the given committee $W$ is a random variable ranging from $n^1$ to $n^1+n^u$. Furthermore, $\SW(W)$ is distributed according to shifted Poisson binomial distribution with $n^u$ independent Bernoulli trials. We first re-label these uncertain approval pairs as $(i_1,c_1),(i_2,c_2),\ldots, (i_{n^u},c_{n^u})$ and the corresponding success probabilities as $(p_1,p_2,\ldots, p_{n^u})$. Intuitively, we represent all the realization of these uncertain approval pairs as a tree as follows.
	\begin{center}
		\begin{forest}
			for tree={
				semithick,
				minimum size=0.2em, inner sep=0pt,
				math content,
				l sep =2mm,
				s sep = 3mm,
				/tikz/arr/.style = { -{Triangle[angle=45:2pt 3]}, shorten >=1pt},
				edge = arr,
			},
			/tikz/ELS/.style = {
				pos=0.5, node
				font=\footnotesize, text=blue, anchor=#1},
			EL/.style={if n=1{edge label={node[ELS=east]{$#1$}}}
				{edge label={node[ELS=west]{$#1$}}}}
			[, phantom, s sep = 0.5cm
			[(i_1\text{, } c_1), tier=0
			[(i_2\text{, } c_2), EL=1-p_1
			[(i_3\text{, } c_3), EL=1-p_2
			[\dots, edge=dashed] 
			[\dots, edge=dashed]
			]
			[(i_3\text{, } c_3), EL=p_2
			[\ldots,edge=dashed]
			[\dots, edge=dashed]
			]
			]
			[(i_2\text{, } c_2), EL=p_1
			[(i_3\text{, } c_3), EL=1-p_2
			[\dots, edge=dashed]
			[\dots, edge=dashed]
			]
			[(i_3\text{, } c_3), EL=p_2
			[, edge=dashed]
			[\dots, edge=dashed]
			]
			]
			]
			[{1}
			[{2}
			[{3}
			[{n^u}, edge=dashed]
			]
			]
			]
			]
		\end{forest}
	\end{center}
	
	Every path in the tree represents a specific realization of uncertain approvals (trials) transforming into approvals (success) or disapprovals (failure). For \SWDist problem $\Prob[\SW(W) = \tau]$, if $\tau < n^1$ or $\tau > n^1+n^u$, $\Prob[\SW(W) = \tau]=0$. So we mainly focus on $\tau \in [n^1,n^1+n^u]$. Denote $t=\tau-n^1$. Then $\Prob[\SW(W)=\tau]$ can be represented as follows.
	\begin{equation}\label{state_trans_equation}
		\begin{aligned}
			& \quad \Prob[\SW(W)=\tau]= \Prob[\SW(W)-n^1 = t] \\
			& =\Prob[t \text{ out of } n^u \text{ trials succeed}] \\
			& =\Big(\Prob[(t-1) \text{ succeed in } (n^u-1) \text{ trials}] \cdot \Prob[n^u\text{-th trial succeeds}]\Big) \\
			&\quad +\Big(\Prob[t \text{ out of } (n^u-1) \text{ trials succeed}]\cdot \Prob[n^u\text{-th trial fails}]\Big).
		\end{aligned}
	\end{equation}
	Based on the above \Cref{state_trans_equation}, we provide the following dynamic programming \Cref{dp_algo_sw_w_equal_t_candidate_probability_model} to solve \SWDist problem.
	\begin{algorithm}[h]
		\caption{\SWDist algorithm for the Candidate Probability model}
		\label{dp_algo_sw_w_equal_t_candidate_probability_model}
		\begin{algorithmic}[1] 
			\REQUIRE $W$, $n^u$,  $t$, $\Delta$.
			\ENSURE $\DP[n^u][t]$.
			\STATE Initialize an $n^u\cdot (t+1)$ matrix $\DP$ with elements $0$;
			\STATE $\DP[1][0] \leftarrow (1-p_1), \DP[1][1] \leftarrow p_1$;
			\FOR{$i \leftarrow (2,\ldots, n^u)$}
			\FOR{$j \leftarrow (1, \ldots, t)$}
			\STATE $\DP[i][j] \leftarrow \big(p_i\cdot \DP[i-1][j-1]\big) + \big((1-p_i)\cdot \DP[i-1][j]\big)$;
			\ENDFOR 
			\ENDFOR 
			\STATE Return $\DP[n^u][t]$.
		\end{algorithmic}
	\end{algorithm}	
	As we re-labeled the $n^u$ uncertain approval pairs (independent Bernoulli trials), in \Cref{dp_algo_sw_w_equal_t_candidate_probability_model}, $\DP[i][j]$ represents the probability that there are $j$ trials which succeed among the first $i$ trials. Then, computing $\Prob[\SW(W)=\tau]$ is equivalent to computing $\DP[n^u][t]$. According to \Cref{state_trans_equation}, $\DP[n^u][t]=p_{n^u}\cdot \DP[n^u-1][t-1] + (1-p_{n^u})\DP[n^u-1][t]$, corresponding to line 5 in \Cref{dp_algo_sw_w_equal_t_candidate_probability_model}. For the computation, we first initialize $\DP[1][1]$ and $\DP[1][0]$, which represent the probability of success or failure of the first trial $(i_1,c_1)$, respectively. This is equal to the probability that voter $i_1$ approves (disapproves) candidate $c_1$, respectively (line 2). From lines 3 to 7, \Cref{dp_algo_sw_w_equal_t_candidate_probability_model} recursively computes $\DP[i][j]$. Finally, \Cref{dp_algo_sw_w_equal_t_candidate_probability_model} returns $\DP[n^u][t]$, which is equal to $\Prob[\SW(W)=\tau]$. We conclude that \SWDist problem is solvable in polynomial time under the Candidate Probability model as \Cref{dp_algo_sw_w_equal_t_candidate_probability_model} runs in $O(mn\cdot kn)=O(kmn^2)$ time.
\end{proof}

\section{Probability of a Committee being Welfare Maximizing}\label{sec:probswm}
With the fundamental tools for computing the social welfare distribution in place, we now turn to the significant problem of \SWMProb, which is to decide, given a committee $W$ and a threshold $p$, whether $\Prob[W \text{ is }\SWM] \geq p$. While we have shown that \SWDist can be computed in polynomial time under both Lottery and Candidate Probability models, the complexity landscape changes significantly when it comes to \SWMProb. In particular, under the Lottery model, we establish the computational intractability of the problem in general, and present intriguing differences between certain special cases. Specifically, for some parameter $\varepsilon > 0$, \textsc{SWM-Prob}$(W, \varepsilon)$ is NP-hard, whereas \isNecSWM can be solved in polynomial time. In contrast, we are able to utilize the \SWDist solution tool to show that \SWMProb is solvable in polynomial time under the Candidate Probability model. 

\begin{theorem}\label{SWM-Prob-Sharp-P-c}
	Under the Lottery model, \SWMProb is NP-hard.
\end{theorem}

\begin{proof}[Proof Sketch]
	Consider the problem of checking whether it is possible for a given committee $W$ to be $\SWM$, i.e., deciding whether $\Prob[W \text{ is } \SWM] > 0$. We first show that there is a polynomial-time reduction from this problem to \textsc{SWM-Prob}$(W, \varepsilon)$ where $\varepsilon < \prod_{i\in [n]}\min_{r\in [s_i]}\{\lambda_r\}$. We next prove that for any committee $W$, deciding whether $\Prob[W \text{ is } \SWM] > 0$ is NP-hard even when $k=1$ by reducing from the Exact Cover by 3-Sets (X3C) problem \citep{GaJo79a}.
\end{proof}

Moreover, there exists a one-to-one correspondence in the reduction construction between each realization under which the given committee $W$ is $\SWM$ and each solution associated with the X3C instance. The next corollary thus follows.
\begin{corollary}\label{coro:swm_prob_lottery_sharp_p_complete}
	Under the Lottery model, given a committee $W$, computing $\Prob[W \text{ is } \SWM]$ is \#P-complete.
\end{corollary}

If we set the parameter $p=1$, the problem becomes deciding whether the given committee is necessarily to be SWM (\isNecSWM problem), for which we obtain a polynomial-time result. The key idea is to carefully construct a deterministic approval profile and prove that any committee $W$ is \textit{necessarily} $\SWM$ if and only if $W$ is $\SWM$ under the constructed deterministic approval profile. Before the formal proof, we introduce the following lemma.

\begin{lemma}\label{IsNecSWM_lemma}
	Under the Lottery model, 
	given a committee $W$, it is a YES instance for  \isNecSWM  if and only if, for every candidate pair $(c,c^\prime)$ where $c\in W$ and $c^\prime \in C\setminus W$, and for every plausible approval profile $\profile$, the approval score of $c$ is at least as large as the approval score of $c'$, that is, $\AS(c, \profile) \geq \AS(c^\prime, \profile)$.
\end{lemma}

With Lemma~\ref{IsNecSWM_lemma} in hand, to prove that \isNecSWM is in P under the Lottery model, it is sufficient to show that it can be checked in polynomial time whether, for all candidate pairs $(c, c^\prime)$ where $c \in W$ and $c^\prime \in C \setminus W$, the condition $\AS(c, \profile) \geq \AS(c^\prime, \profile)$ holds for every plausible approval profile $\profile$.

\begin{theorem}\label{thm:IsNecSWM-Lot}
	Under the Lottery model, \isNecSWM is solvable in polynomial time.
\end{theorem}

\begin{proof}[Proof Sketch]
	For each candidate pair $(c, c^\prime)$, we construct a \textit{deterministic} approval profile $\bar{\profile}$ and demonstrate that if $(c, c^\prime)$ satisfies $\AS(c,\bar{\profile}) \geq \AS(c^\prime,\bar{\profile})$ for the constructed profile $\bar{\profile}$, then $(c, c^\prime)$ satisfies $\AS(c,\profile) \geq \AS(c^\prime,\profile)$ for all plausible approval profiles $\profile$.
	The construction is as follows. Given a committee $W$, for each pair $(c, c^\prime)$ where $c \in W$ and $c^\prime \in C \setminus W$ and each voter’s approval set $A_i$, there are four possible cases: (1) $c^\prime \in A_i$ and $c \notin A_i$; (2) $c^\prime \in A_i$ and $c \in A_i$; (3) $c^\prime \notin A_i$ and $c \notin A_i$; (4) $c^\prime \notin A_i$ and $c \in A_i$. We construct the deterministic profile $\bar{\profile}$ as follows: for each voter $i$, set $\bar{A_i}$ by selecting a plausible approval set in the following priority order: (1) $\succ$ (2) $\succ$ (3) $\succ$ (4). That is, we first check whether there exists a plausible approval set such that $c^\prime$ is in the approval set while $c$ is not. If such an approval set exists, we set it as $\bar{A}_i$ in the deterministic approval profile $\bar{\profile}$; otherwise, we consider cases (2), (3), and (4) in sequence.	
	
	Next, we prove that if a pair $(c, c^\prime)$ satisfies $\AS(c,\bar{\profile}) \geq \AS(c^\prime,\bar{\profile})$ in the constructed profile $\bar{\profile}$, then the pair $(c, c^\prime)$ satisfies $\AS(c, \profile) \geq \AS(c^\prime, \profile)$ for all plausible approval profiles $\profile$.
	
	Since the construction and verification can be computed in polynomial time, \isNecSWM under the Lottery model is solvable in polynomial time.
\end{proof}

We now turn into the Candidate Probability model and show that, under this model, the \SWMProb problem is solvable in polynomial time by showing $\Prob[W \text{ is }\SWM]$ is polynomial-time computable.
\begin{theorem}\label{thm:MaxSW-Prob-CP-3VA}
	Under the Candidate Probability model, the problem \SWMProb is solvable in polynomial time.
\end{theorem}

\begin{proof}
	Given a committee $W$, we may re-label the candidates, without loss of generality, so that 
	$W=\{c_1,\cdots, c_k\}$. Let $\AS(c)$ denote the random variable corresponding to the approval score of candidate $c\in C$. 
	Under the Candidate Probability model, the probability that a committee $W$ is $\SWM$ is equivalent to the probability of sampling an approval profile where the approval scores of the candidates in $W = \{c_1, \cdots, c_k\}$ rank among the top-$k$.
	\begin{equation*}\label{maxSW-Prob-equation}
		\begin{split}
			\Prob\left[W \!\text{ is }\!\SWM \right] 
			&=\! \sum_{\mathcal{A}\in \mathbf{A}}\!\!\Delta(\mathcal{A})\!\cdot\! \mathbb{I}[\text{\sc IsSWM}(W, \mathcal{A})]\\
			&=\! \sum_{\mathcal{A}\in \mathbf{A}}\!\! \Delta(\mathcal{A})\!\cdot\! \mathbb{I}[\AS(c_1,\profile),\cdots\!, \AS(c_k,\profile) \text{ rank top-}k]\\
			&=\!\Prob\left[\max\limits_{c\in C\setminus W} \{\AS(c) \}  \leq \min_{1\leq i\leq k}\{\AS(c_i)\}\right].
		\end{split}
	\end{equation*}
	Conditioning on the value of $\min_{1\leq i\leq k} \AS(c_i)$, the probability $\Prob\left[\max\limits_{c\in C\setminus W} \{\AS(c) \}  \leq \min_{1\leq i\leq k}\{\AS(c_i)\}\right]$ is rewritten as
	\begin{align*}
		\sum_{t=0}^n\! \Prob\left[\!\min_{1\leq i\leq k}\!\{\AS(c_i)\}\! =\! t\right]\! \cdot \! \Prob\!\left[\!\max\limits_{c\in C\setminus W} \! \{\AS(c) \} \!\leq\! t \!\mid \! \min_{1\leq i\leq k}\! \{\AS(c_i)\}\! =\! t\right]\!.
	\end{align*}
	For any two candidates $c_i, c_j \in C$, $\AS(c_i)$ is independent of $\AS(c_j)$ because for each voter $v\in V$, the event that $v$ approves $c_i$ is independent of the event where $v$ approves $c_j$. Notably, for the conditional probability $\Prob[\min_{1\leq i\leq k}\{\AS(c_i)\} ] = t$, the random variable $\min_{1\leq i\leq k}\{\AS(c_i)\}$ only depends on $\{\AS(c_1),\cdots, \AS(c_k)\}$ and is independent of $\{\AS(c_{k+1}),\cdots, \AS(c_{m})\}$. It follows that
	\begin{align*}\label{maxSW-Prob-equation-condition-prob}
		\Prob\left[\max\limits_{c\in C\setminus W} \{\AS(c) \}  \leq \min_{1\leq i\leq k}\{\AS(c_i)\}\right]=\sum_{t=0}^{n} \Prob\left[\min_{1\leq i\leq k}\{\AS(c_i)\} = t\right]  \cdot \left( \prod_{c\in C\setminus W} \Prob\left[\AS(c) \leq t \right]\right).
	\end{align*}
	Now it boils down to compute $\Prob\left[\min_{1\leq i\leq k}\{\AS(c_i)\} = t \right]$, which can be reformulated as 
	\begin{align*}
	\Prob\left[\min_{1\leq i\leq k}\{\AS(c_i)\} = t \right]
		&= \Prob \left[\min_{1\leq i\leq k}\{\AS(c_i)\} > (t - 1) \right]  - \Prob \left[\min_{1\leq i\leq k}\{\AS(c_i)\} > t\right] \\
		&= \prod_{1\leq i\leq k} \Prob \left[\AS(c_i) > (t - 1) \right]    -  \prod_{1\leq i\leq k} \Prob \left[\AS(c_i) > t\right] \\
		&= \prod_{c_i \in W} \left( \sum_{r=t}^{n}  \Prob \left[\AS(c_i) =r\right] \right) -  \prod_{c_i \in W} \left( \sum_{r=t + 1}^{n}  \Prob \left[\AS(c_i) =r\right]  \right). 
	\end{align*}
	
	Since the \SWDist problem under the Candidate Probability model can be solved in polynomial time (see \Cref{thm:SWDist-CP}), computing $\Prob \left[\AS(c_i) = r\right]$ can be done in  polynomial time, implying that \SWMProb under the Candidate Probability model is solvable in polynomial time.
\end{proof}

\section{Committees that Maximize $\SWM$ Probability }\label{sec:maxswm}
We now turn our attention to the problem of \maxSWM, which asks whether there exists a committee $W$ such that $W$ is $\SWM$ with probability at least $p$.
Our first result is that under the Lottery model, \maxSWM is NP-hard, even in the instance with single voter. 

\begin{theorem}\label{th:MaxSWM-Lottery}
	Under the Lottery model, \maxSWM is NP-hard, even when $n=1$.
\end{theorem}

\begin{proof}[Proof Sketch]
	We prove that the problem is NP-hard via a reduction from the {\sc Min-$r$-Union} (M$r$U) problem \citep{VINT02a}. In M$r$U, we are given a universe $U$ of $m$ elements, a collection $\mathcal{S}\subseteq 2^U$ of $q$ sets, and two integers $r\leq q$ and $\ell$. The goal is to decide whether there exists a sub-collection $I \subseteq [q]$ with size $r$ such that $|\bigcup_{i\in I} S_i| \leq \ell$. In the reduction, our construction maps a YES M$r$U instance 
	to the existence of a committee $W$ of size $\ell$ with $\Prob[W \text{ is } \SWM] \geq \frac{r}{q}$, and vice versa.
\end{proof}

\Cref{th:MaxSWM-Lottery} implies that \maxSWM is NP-hard also for the Joint Probability model, as both uncertainty models coincide in single-voter instances. In view of this computational intractability, we consider the special case when $p=1$, i.e., the \existsNecSWM problem, which involves determining whether there exists a committee $W$ that is $\SWM$ with probability $1$. For ease of clarity, we say a committee $W^\ast$ is \textit{necessarily} $\SWM$ if $W^\ast$ is $\SWM$ with probability $1$.

We present \Cref{algo:ExistsNecSWM-Lot} to show that \existsNecSWM is solvable in polynomial time.
The algorithm is built upon the ``dominance graph". Specifically, for any candidate pair $(c_i,c_j) \in C$, we say that $c_i$ ``dominates" $c_j$ if for every plausible approval profile $\profile$, $\AS(c_i,\profile) \geq \AS(c_j,\profile)$, denoted by $c_i \succeq^{\AS} c_j$. 
The dominance relation between any candidate pair $(c_i,c_j)$ can be verified in polynomial time by constructing the \textit{deterministic} approval profile $\bar{\profile}$ established in \Cref{thm:IsNecSWM-Lot}.

By enumerating all candidate pairs, we construct a dominance digraph $G=(C,E)$ where $C$ is the set of candidates and $E$ contains directed edge representing dominance relations. Specifically, each edge $(c_i,c_j)\in E$ indicates $c_i \succeq^{\AS} c_j$. In case $\AS(c_i,\profile) = \AS(c_j, \profile)$, we break ties by lexicographic order, treating $c_i$ as dominating $c_j$. No edge is added between $c_i$ and $c_j$ if neither dominates the other. Note that the dominance relation is transitive, and thus the digraph $G$ is acyclic. With the dominance graph constructed, we now present an auxiliary lemma that underpins the design of \Cref{algo:ExistsNecSWM-Lot}.
\begin{lemma}
	\label{lemma:exist_nec_swm_no_arc}
	Given a dominance graph $G=(C,E)$, for any $c_i,c_j\in C$ with no edge between $c_i$ and $c_j$, any necessarily SWM committee $W^\ast$ satisfies either $\{c_i,c_j\} \subseteq W^\ast$ or $\{c_i,c_j\}\cap W^\ast=\emptyset$.
\end{lemma}

Lemma~\ref{lemma:exist_nec_swm_no_arc} establishes that, for each candidate pair $c_i,c_j$ without dominance relation, if a necessarily $\SWM$ committee exists, then $c_i$ and $c_j$ must either both be included or both be excluded. Leveraging this property, we now present \Cref{algo:ExistsNecSWM-Lot} to solve the \existsNecSWM problem.

\begin{algorithm}[h]
	\caption{\existsNecSWM algorithm for the Lottery model}
	\label{algo:ExistsNecSWM-Lot}
	\begin{algorithmic}[1] 
		\REQUIRE $G =(C, E), k$. \\
		\ENSURE YES or NO.
		\STATE Initialize $W \leftarrow \emptyset$ and  $\bar{G}\leftarrow(\bar{C}=C,\bar{E}=E)$;
		\WHILE{$|W| \leq k$}
		\STATE Select a candidate $c^\ast$ with zero indegree in $\bar{G}$ (breaking ties arbitrarily); 
		\STATE Add $c^\ast$ into $W$;
		\STATE Update $\bar{G}$ by deleting all the edges in $\{(c^\ast, c^\prime), c^\prime \in C\setminus \{c^\ast\}: (c^\ast, c^\prime)\in \bar{E}\}$;
		\ENDWHILE
		\STATE Return YES if $\forall~ c\in W, c^\prime \in C\setminus W: (c,c^\prime) \in E$ otherwise NO;
	\end{algorithmic}
\end{algorithm}

\Cref{algo:ExistsNecSWM-Lot} takes the dominance graph $G$ as input. First, it initializes an empty candidate set $W$ and creates a duplicate of $G$, denoted as $\bar{G}$. The algorithm then iteratively selects a candidate $c^\ast$ with zero indegree, adds it to $W$, and updates $\bar{G}$ by removing the outgoing edges from $c^\ast$ in each round (lines 2-6). Since $G$ is acyclic, a node with zero indegree is guaranteed to exist in the first iteration of the while loop. In each subsequent iteration, $\bar{G}$ remains acyclic as it is a subgraph of $G$, ensuring that there is always a node with zero indegree, which implies that the algorithm terminates. 

\begin{theorem}\label{thm:ExistsNecSWM-Lot}
	Under the Lottery model, \existsNecSWM is solvable in polynomial time.
\end{theorem}

For Candidate Probability model, identifying the complexity of \maxSWM problem becomes substantially more challenging, even under the restrictive 3VA setting. A natural initial hypothesis is that the expected $\SWM$ committee also maximizes the probability of being $\SWM$. Unfortunately, this intuition does not hold in general. We provide a counterexample involving $3$ voters and $4$ candidates to illustrate the inherent difficulty.
\[
\bordermatrix{ & 1 & 2 & 3 & 4\cr
	1 & $0.5$ & $1.0$ & $1.0$ & $1.0$ \cr
	2 & $0.5$ & $0.5$ & $1.0$ & $0.5$ \cr
	3 & $0.5$ & $0.0$ & $0.0$ & $0.0$} \qquad
\]
Each element $p_{i,c}$ is the probability that voter $i$ approves candidate $c$. We first compute that committees $W_1=\{1,3\}$, $W_2=\{2,3\}$ and $W_3=\{3,4\}$ all achieve the highest expected social welfare, with a value of $2.5$. However, their probability of being $\SWM$ differ: $\Prob \left[W_1 \text{ is } \SWM \right]=\frac{19}{32}$, $\Prob \left[W_2 \text{ is }\SWM\right]=\frac{18}{32}$, and $\Prob \left[W_3 \text{ is }\SWM \right]=\frac{18}{32}$. This example highlights the inherent complexity and subtlety of the \maxSWM problem, even under the 3VA model. Despite these challenges, we prove that it is polynomial-time solvable when $n=1$ or $k$ is constant (See details in the appendix). Beyond these aforementioned results under restrictive assumptions, we also establish a positive result for \existsNecSWM problem.
\begin{theorem}\label{thm:ExistsNecSWM-CP-3VA}
	Under the Candidate Probability model, \existsNecSWM is solvable in polynomial time.
\end{theorem}

The core idea is an poly-time algorithm which constructs a profile consisting solely of certain approval ballots and computes the $\SWM$ committee for this deterministic approval profile\footnote{Tie-breaking by selecting the committee with the highest number of approvals with positive probabilities.}. The algorithm returns YES if the committee is a certificate for YES instance for \isNecSWM, otherwise returns NO.
It is worth mentioning that for both Lottery and Candidate Probability model,
not only can we determine the \existsNecSWM problem in polynomial time, but we can also compute the committee maximizing the probability of being $\SWM$ by our proposed polynomial-time algorithms. 
	
	
\section{Robust Committees}\label{sec:robust}

In many settings, not only deciding the \ExistsSWMProb problem is intractable, but it can also be highly sensitive to the realizations of the approval profiles. In particular, it might be the case that a committee that maximizes probability of being $\SWM$ may only be social welfare maximizing for only a small fraction of the plausible profiles while performing poorly in the remaining plausible profiles. As a result, this motivates us to study \textit{robust} committees. Intuitively, a committee is considered robust if it achieves approximately optimal social welfare with high probability.
\begin{definition}[$(\alpha,\beta)$-Robust Committee]
	\label{def:roubst_committee}
	Given any uncertain preference model, denote by $W^*$ the \textit{random committee} that maximizes social welfare under any plausible approval profile and a \textbf{random variable} $Z$ representing its social welfare.
	A committee $W$  is $(\alpha,\beta)$-robust if it satisfies $$\Prob\left[\SW(W) \geq \alpha \cdot Z \right]\geq \beta, \text{ where } \alpha, \beta \in (0,1].$$
\end{definition}
Although  we know that the a committee maximizing the expected social welfare may not be the committee with highest probability of being $\SWM$, the following result shows that it is guaranteed to be  $(\frac{1}{2},\frac{1}{2})$-robust.

\begin{theorem}\label{eswm_committee_robust_3va_model}
	Under the 3VA model, any committee $W$ maximizing the expected social welfare is $(\frac{1}{2},\frac{1}{2})$-robust.
\end{theorem}

\begin{proof}
	Let $W$ be a committee that maximizes the expected social welfare. To prove $W$ is $(\frac{1}{2}, \frac{1}{2})$-robust, we aim to show that $\Prob[\SW(W) \geq \frac{1}{2}\cdot Z] \geq \frac{1}{2}$, where $Z$ is the random variable representing the social welfare of the social welfare maximizing committee $W^*$. We next introduce the committee $\bar{W}$, 
	denoting the committee that has the maximum number of non-zero approval entries (that is, the number of pairs $(i,c)$ for $i\in N$ and $c\in C$ such that $p_{i,c}\neq 0$). 
	
	For any plausible approval profile $\mathcal{A}$ under the 3VA model, we first observe that the social welfare achieved by $W^*$ is at most the number of non-zero approval entries for any $k$ candidates in $C$. Moreover, recall the definition of $\bar{W}$ and we know that the number of non-zero approval entries for any $k$ candidates in $C$ is upper-bounded by $2\cdot \mathbb{E}[\SW(\bar{W})]$ since $\bar{W}$ has the highest number of non-zero approval entries and $\mathbb{E}[\SW(\bar{W})]$ is equal to the sum value of all approval entries for $\bar{W}$ and the equality holds for the upper bound when all the non-zero approval entries for committee $\bar{W}$ are $\frac{1}{2}$. That is, $\forall~c\in \bar{W}$, $\forall~i\in N$ such that $p_{i,c}\neq 0$, then $p_{i,c}=\frac{1}{2}$. Consequently, we have $\SW(W^*) \leq 2\cdot \mathbb{E}[\SW(\bar{W})]$ for any plausible approval profile $\profile$. 
	On the other hand, by the definition of $W$, we have $\mathbb{E}[\SW(W)] \geq \mathbb{E}[\SW(\bar{W})]$. Hence, we bound the probability of $\SW(W) \geq \frac{1}{2}\cdot Z$ by 
	\begin{align*}
		\Prob[\SW(W) \geq \frac{1}{2}\cdot Z] &= \Prob[\SW(W) \geq \frac{1}{2}\cdot \SW(W^*)]\\
		&\geq \Prob[\SW(W) \geq \mathbb{E}[\SW(\bar{W})]] \\
		& \geq \Prob[\SW(W) \geq \mathbb{E}[\SW({W})]]=\frac{1}{2}.
	\end{align*} 
	The last step holds as $\SW(W)$ follows a shifted binomial distribution, i.e.,  $\SW(W) \sim \text{Bin}(y, \frac{1}{2}) + x$, thus the probability that $\SW(W)$ is larger than its expectation is $\frac{1}{2}$. Therefore, for any expected social welfare maximization committee $W$ under the 3VA model, we have $\Prob\left[\SW(W) \geq \frac{1}{2}\cdot Z\right] \geq \frac{1}{2}$, i.e., $W$ is $(\frac{1}{2}, \frac{1}{2})$-robust.
\end{proof}

Since a committee maximizing the expected social welfare can be computed in polynomial time, \Cref{eswm_committee_robust_3va_model} shows that $(\frac{1}{2}, \frac{1}{2})$-robust committees can be efficiently computed under the 3VA model. However, our next result shows that a similar result does not hold in the more general setting of the Candidate Probability model.
\begin{theorem}\label{no_robust_committee_candidate_probability_model}
	Under the Candidate Probability model, for $\alpha,\beta\in (0,1]$, no committee is  $(\alpha,\beta)$-robust.
\end{theorem}

\begin{proof}
	Given any $ \alpha, \beta\in (0,1]$, we construct an instance under which no committee is $(\alpha,\beta)$-robust.
	Consider an instance with single voter. The sole voter independently approves each candidate with identical probability $p = \frac{\beta}{2}$. Consider the candidate number $m$ such that $m> \frac{\log(p)}{\log(1-p)}$ and set the committee size to be $k=1$. 
	
	Since all candidates have the same approval probability and the committee size is $1$, any committee will have the same performance in terms of the probability of being $\SWM$. Without loss of generality, let $W=\{c_1\}$. Recall the definition of random variable $Z$ which represents the social welfare of the $\SWM$ committee. In the constructed instance, we express $Z$ as $Z=\max_{c_i\in C}\{\AS(c_i)\}$. For committee $W$, we have
	\begin{align}\label{eq:robust1}
		\Prob\left[ \SW(W)\! \geq\! \alpha \!\cdot\! Z\right]
		&=\Prob\left[ \AS(c_1) \geq \alpha \cdot \max_{c_i\in C\setminus\{c_1\}}\{\AS(c_i)\}\right].
	\end{align}
	Now consider $\Prob\left[\max_{c_i\in C\setminus\{c_1\}}\{\AS(c_i)\} = 1\right]$, we obtain 
	\begin{align}\label{eq:robust2}
		\Prob\left[\max_{c_i\in C\setminus\{c_1\}}\{\AS(c_i)\} = 1 \right]
		&= 1 - \Prob\left[~\forall\, c_i\in C\setminus\{c_1\}, \AS(c_i)=0 \right]\notag \\
		&= 1 - (1 - p)^{m - 1}.
	\end{align}
	Let $Y$ be $\max_{c_i\in C\setminus\{c_1\}}\{\AS(c_i)\}$.  Conditioning on $Y$, we have 
	\begin{align*}
		\Prob\!\left[\! \SW(W)\! \geq\! \alpha \!\cdot\! Z\right]
		& = \Prob\left[ \AS(c_1) \geq \alpha \cdot Y \right] \tag{\Cref{eq:robust1}} \\
		&= \Big(\Prob\left[\AS(c_1) \geq 0 \mid Y = 0 \right] \cdot \Prob\left[Y = 0\right]\Big) \\
		&\quad + \Big(\Prob\left[\AS(c_1) \geq \alpha \mid Y = 1 \right] \cdot \Prob\left[Y = 1\right]\Big)\\
		&= \Prob\left[Y = 0\right] + \Prob\left[\AS(c_1) \geq \alpha \right]\cdot \Prob\left[Y=1\right] \tag{Independence} \\
		&= (1-p)^{m-1} \!+\! \Prob\left[ \AS(c_1) \geq \alpha \right] \!\cdot\! \left(1 - (1 - p)^{m - 1}\right) \tag{\Cref{eq:robust2}}\\
		&=  (1-p)^{m-1} + p \cdot \left(1 - (1 - p)^{m - 1}\right) \tag{With probability $p$ approving $c_1$} \\
		&= p + (1 - p)^m < \beta. \tag{$m > \frac{\log(p)}{\log(1-p)}$}  
	\end{align*}
	This implies that there is no committee satisfying $(\alpha, \beta)$-robust in this instance.
\end{proof}
Since every candidate probability instance admits a unique representation as a lottery model (Proposition~\ref{unique_joint_prob_for_lottery}). Hence, the impossibility result applies for the lottery and joint probability model.

	
\section{Conclusions}
This paper initiates the study of social welfare maximization under uncertainty in approval-based committee voting. Given the relevance of such voting rules in applications like recommender systems and block-chain governance, our framework provides a foundation for incorporating uncertainty into these domains. Many of our results include explicit polynomial-time algorithms, offering a general toolkit for addressing a broader class of set selection problems under uncertainty.
Our analysis focuses on social welfare defined by the size of the intersection between the selected committee and voters’ approval sets. It would be interesting to investigate whether some of our techniques extend to more general satisfaction functions. A particularly compelling open question is the computational complexity of \maxSWM under the candidate probability model, even in the restricted 3-Valued Approval (3VA) setting. Since \maxSWM is NP-hard in other models, a promising direction for future work is the development of approximation algorithms for this problem.


\section*{Acknowledgment}
	This work was partially supported by the NSF-CSIRO grant on “Fair Sequential Collective Decision-Making” (Grant No. RG230833) and the ARC Laureate Project FL200100204 on “Trustworthy AI.” Venkateswara Rao Kagita was funded by the Science and Engineering Research Board (SERB) through the SIRE Fellowship (Project ID SIR/2022/001217).
	
	\bibliographystyle{abbrvnat}
	\bibliography{sample}


\clearpage
\appendix

Omitted results and proofs are provided in the appendix as follows.

\section{\maxExpSW: Expected Social Welfare Maximization Committee}\label{sec:maxExpSW}

\maxExpSW is the problem of computing a committee that maximizes the expected social welfare ($\E[\SW(W)]$). Formally, $\E[\SW(W)]$ of a committee $W$ is defined as follows. 
\[
\begin{aligned}
	\E[\SW(W)] &= \sum_{\profile \in \profileset} \Delta(\profile) \cdot \SW(W, \profile)\\
	&= \sum_{\profile \in \profileset} \Delta(\profile) \cdot \sum_{i\in V}|W\cap A_i|.
\end{aligned}
\]

\begin{theorem}\label{thm:MaxExpSW-4models}
	For every uncertain preference model, \maxExpSW is solvable in polynomial time.
\end{theorem}

\begin{proof}
	\noindent\textbf{(Joint Probability Model)}  \maxExpSW problem under the joint probability model can be represented as determining the the committee $W^\ast$ such that 
	\[
	\begin{aligned}
		W^\ast &= \underset{W \subseteq C}{\arg\max}~ \E[\SW(W)] \\
		&= \underset{W \subseteq C}{\arg\max} \sum_{\profile \in \profileset}  \Delta(\profile) \cdot \SW(W, \profile)\\
		&= \underset{W \subseteq C}{\arg\max} \sum_{\profile \in \profileset}  \Delta(\profile) \cdot \left( \sum_{c\in W}\AS(c,\profile) \right)\\
		&= \underset{W \subseteq C}{\arg\max} \sum_{c \in W}  \left(\sum_{\profile \in \profileset}  \Delta(\profile) \cdot \AS(c,\profile)\right).
	\end{aligned}
	\]
	To maximize $\sum_{c \in W}  \left(\sum_{\profile \in \profileset}  \Delta(\profile) \cdot \AS(c,\profile)\right)$, we can enumerate all the candidates $c\in C$ and compute the top-$k$ candidates maximizing $\left(\sum_{\profile \in \profileset}  \Delta(\profile) \cdot \AS(c,\profile)\right)$ by iterating over all plausible approval profiles $\profile$ and summing the products of $\Delta(\profile)$ and $\AS(c,\profile)$ in polynomial time.
	
	\textbf{(Lottery Model)} \maxExpSW problem under the lottery model can be written as
	\[
	\begin{aligned}
		W^\ast &= \underset{W \subseteq C}{\arg\max}~ \E[\SW(W)] \\
		&= \underset{W \subseteq C}{\arg\max} \sum_{i\in V}  \sum_{\profile \in \profileset} \Delta(\profile) \cdot |W\cap A_i|\\
		&= \underset{W \subseteq C}{\arg\max} \sum_{i\in V}  \sum_{\profile \in \profileset}\Delta(\profile) \cdot \left(\sum_{c\in W}1\cdot \mathbb{I}[c\in A_i]\right)\\
		&= \underset{W \subseteq C}{\arg\max} \sum_{i\in V} \sum_{c\in W} \left( \sum_{r\in [s_i]}\lambda_r \cdot \mathbb{I}[c \in S_r]\right)\\
		&= \underset{W \subseteq C}{\arg\max} \sum_{c\in W} \left(\sum_{i\in V}\sum_{r\in [s_i]}\lambda_r \cdot \mathbb{I}[c \in S_r]\right).
	\end{aligned}
	\]
	Recall that $s_i$ denotes the set of plausible approval sets for voter $i$. Here, $\mathbb{I}[c \in S_r]$ is an indicator function which returns $1$ if candidate $c$ belongs to the plausible approval set $S_r$ of voter $i$. To maximize $\E[\SW(W)]$, we can enumerate all the candidates $c\in C$ and choose the top-$k$ candidates who maximize $\sum_{i\in V}\sum_{r\in [s_i]}\lambda_r \cdot \mathbb{I}[c \in S_r]$, which is polynomial-time computable. 
	
	\textbf{(Candidate Probability Model)} \maxExpSW problem under the candidate probability model can be formalized as
	\[
	\begin{aligned}
		W^\ast &= \underset{W \subseteq C}{\arg\max}~ \E[\SW(W)] \\
		&= \underset{W \subseteq C}{\arg\max} \sum_{i\in V} \sum_{\profile \in \profileset} \Delta(\profile) \cdot |W\cap A_i|\\
		&= \underset{W \subseteq C}{\arg\max} \sum_{i\in V}  \sum_{\profile \in \profileset}\Delta(\profile) \cdot \left(\sum_{c\in W}1\cdot \mathbb{I}[c\in A_i]\right)\\
		&= \underset{W \subseteq C}{\arg\max} \sum_{i\in V} \sum_{c\in W} 1\cdot \Prob[c\in A_i] \\
		&= \underset{W \subseteq C}{\arg\max} \sum_{c\in W} \sum_{i\in V} p_{i,c}.
	\end{aligned}
	\]
	To maximize $\E[\SW(W)]$ under the candidate probability model, we can iterate every candidate $c\in C$ and sum up the probabilities $p_{i,c}$ for all voters $i\in V$. The committee maximizing the expected social welfare consists of the candidates who rank top-$k$ with the maximum value of $\sum_{i\in V} p_{i,c}$.
\end{proof}

\section{Omitted Results from Section 3}

\subsection{Omitted Example for \Cref{thm:SWDist-lot}}\label{example:SWDist_lot}
We provide the following example showing how to compute the social welfare distribution for a given committee under the lottery model via the dynamic programming in \Cref{dp_algo_sw_w_equal_t_lottery_model}.
\begin{example}[Demonstration of \Cref{dp_algo_sw_w_equal_t_lottery_model} under the Lottery model]
	Consider $V=\{1,2,3\}, C=\{1,2,3\}$, $W=\{2,3\}$ and $\tau=3$. The Lottery preference profile is given as follows:
	\[
	\begin{aligned} 
		&\text{Voter 1}: \{ (0.3, \{1,2\}); (0.5, \{2,3\}); (0.2, \{1,2,3\})\} \\
		&\text{Voter 2}: \{ (0.4, \{1,2\}); (0.6, \{3\})\} \\
		&\text{Voter 3}: \{ (0.5, \{1\}); (0.1, \{1,3\}); (0.4, \{2,3\})\} \\
	\end{aligned}
	\]
	Firstly, we preprocess the computation of the ``contribution" for each voter $1,2,3$. The result is listed in Table \ref{dynamic_programming_preprocessing_lottery_model}.
	\begin{table}[!htbp]
		\caption{Computation of the $f_{n\cdot (k+1)}$ matrix}
		\label{dynamic_programming_preprocessing_lottery_model}
		\centering
		\begin{tabular}{{p{1cm} @{\qquad} p{1cm} @{\qquad} p{1cm} @{\qquad} p{1cm}}}
			\toprule
			$f[i][j]$      &   $0$    &  $1$    & $2$  \\ \midrule
			\text{Voter 1} &   $0.0$  &  $0.3$  & $0.7$  \\ \midrule
			\text{Voter 2} &   $0.0$  &  $1.0$  &  $0.0$  \\ \midrule
			\text{Voter 3} &   $0.5$  &  $0.1$  &  $0.4$  \\ \bottomrule
		\end{tabular}
	\end{table} 
	For each row, we compute the probability that each voter ``contributes" $0, 1, 2$ social welfare under the given $W=\{2,3\}$. For example, for voter $1$, $f[1][1]=0.3$ represents the 
	the probability voter $1$ ``contributes" $1$ to $\SW(W)$ is $0.3$. This is because voter $1$ contributes $1$ to the social welfare of $W=\{2,3\}$ only when her approval set realization is $\{1,2\}$. For the other two possible approval sets $\{2,3\}$ and $\{1,2,3\}$, they both contribute $2$ to $\SW(W)$, thus we have $f[1][2]=0.5 + 0.2 = 0.7$.
	
	After the preprocessing, we do the dynamic programming procedures, starting from initialization with only voter $1$. Table \ref{dynamic_programming_procedures_lottery_model} shows the results.
	\begin{table}[h]
		\caption{Computation of $\DP_{n\cdot (nk+1)}$ matrix}
		\label{dynamic_programming_procedures_lottery_model}
		\centering
		\scalebox{1}{
			\begin{tabular}{p{1cm} @{\qquad} p{.4cm} @{\qquad} p{.4cm} @{\qquad} p{.4cm}@{\qquad}p{.4cm} @{\qquad} p{.4cm} @{\qquad} p{.4cm} @{\qquad} p{.4cm}}
				\toprule
				$\DP[i][j]$   &       $0$   & $1$     & $2$     & $3$    & $4$    & $5$    & $6$   \\ \midrule
				$\{1\}$                   &       $0.0$   &  $0.3$  & $0.7$ & $0.0$ & $0.0$ & $0.0$  & $0.0$ \\ \midrule
				$\{1, 2\}$               &       $0.0$   &  $0.0$  & $0.3$ & $0.7$ & $0.0$ & $0.0$  & $0.0$ \\ \midrule
				$\{1,2,3\}$             &       $0.0$   &  $0.0$  & $0.15$ & $\mathbf{0.38}$ & $0.19$ & $0.28$  & $0.0$ \\ \bottomrule
		\end{tabular}}
	\end{table} 
	We first get $\DP[1][0]=f[1][0]=0,\DP[1][1]=f[1][1]=0.3,\DP[1][2]=0.7$. Next, taking voter $2$ into consideration, $\DP[2][0]=\DP[1][0]\cdot f[2][0]=0$, $\DP[2][1]=\DP[1][0]\cdot f[2][1] + \DP[1][1]\cdot f[2][0]=0\cdot 1 + 0.3\cdot 0=0$. Similarly, we get $\DP[2][2]=0.3$ and $\DP[2][3]=0.7$. For our target $\Prob\left[\SW(W) = 3\right]$, i.e., $\DP[3][3]$. After the computation of $\DP[2][t]$ for $t$ from $0$ to $6$, we compute $\DP[3][3]$ as follows.
	\begin{align*}
		\DP[3][3]
		&=\sum_{r=0}^3 \DP[2][r] \cdot f[3][3-r] \\
		&= 0 \cdot 0 + 0\cdot 0.4 + 0.3 \cdot 0.1 + 0.7 \cdot 0.5 \\
		&= 0.03 + 0.35 = 0.38.
	\end{align*}
	Hence, the probability $\Prob[\SW(W)=3]$ is $0.38$.
\end{example}

\subsection{Omitted Example for \Cref{thm:SWDist-CP}}\label{example:SWDist_CP}
To further illustrate \Cref{dp_algo_sw_w_equal_t_candidate_probability_model}, we present the following example demonstrating how the algorithm solves the \SWDist problem under the Candidate Probability model.
\begin{example}[Demonstration of \Cref{dp_algo_sw_w_equal_t_candidate_probability_model} under the Candidate Probability model]	
	Consider $V = \{1,2\}$, $W = \{1,2\}$, and $\tau = 3$. The Candidate Probability preference profile is represented as follows.
	\[
	\bordermatrix{ & 1 & 2\cr
		1 & $1.0$   & $0.5$ \cr
		2 & $0.6$ & $0.8$ } \qquad
	\]
	
	We first compute $n^1=1$ (voter $1$ certainly approves candidate $1$) and $n^u=3$ ($p_{1,2}= 0.5$, $p_{2,1} = 0.8$, and $p_{2,2} = 0.6$). Then, we re-label these three pairs of uncertain approvals as $(1,2),(2,1),(2,2)$ with probabilities $p_1=0.5,p_2=0.8,p_3=0.6$. To compute $\Prob[\SW(W)=3]=\Prob[\SW(W)-n^1=3-n^1]$, it is to compute the probability of two successful trials out of these three independent Bernoulli trials, i.e., $\DP[3][2]$. We first compute $\DP[1][0]=0.5,\DP[1][1]=0.5$. To compute $\DP[3][2]$, it can be represented as
	\begin{align*}
		\DP[3][2] &= p_3 \cdot \DP[2][1] + (1-p_3)\DP[2][2]\\
		&= p_3 \cdot \Big(p_2\cdot \DP[1][0]+(1-p_2) \cdot \DP[1][1]\Big) \\
		&\quad  + (1-p_3)\cdot \Big(p_2\cdot \DP[1][1] + (1-p_2)\cdot \DP[1][2]\Big)\\
		&= 0.6 \cdot (0.8\cdot 0.5 + 0.2\cdot 0.5) + 0.4 \cdot (0.8\cdot 0.5 + 0.2\cdot 0)\\
		&=0.6 \cdot 0.5 +0.4\cdot 0.4 = 0.46.
	\end{align*}
	
	Therefore, we get the solution that $\Prob[\SW(W)=3]=\DP[3][2]=0.46$.
\end{example}

\section{Omitted Proofs from Section 4}
\subsection{Proof of \Cref{SWM-Prob-Sharp-P-c}}

\begin{proof}
	Given a committee $W$, we define the problem of deciding whether $W$ is possible to be social welfare maximizing under some plausible approval profile as \isPossSWM. Consider the threshold $\varepsilon < \prod_{i\in [n]}\min_{r\in [s_i]}\{\lambda_r\}$ and the problem \textsc{SWM-Prob}$(W, \varepsilon)$. We first show that \isPossSWM is a YES instance \textit{if and only if} \textsc{SWM-Prob}$(W, \varepsilon)$ is a YES instance. The ``if" direction is immediate. For the reverse direction, consider any YES instance in \isPossSWM. It follows that $W$ is $\SWM$ under at least one plausible approval profile under the Lottery model. Notice that every plausible approval profile under the Lottery model has at least $\prod_{i\in [n]}\min_{r\in [s_i]}\{\lambda_r\}$ realization probability. Then for the given committee $W$, $\Prob[W \text{ is } \SWM] \geq \prod_{i\in [n]}\min_{r\in [s_i]}\{\lambda_r\}$, implying it is a YES instance in \textsc{SWM-Prob}$(W, \varepsilon)$. 
	
	We next prove that the \isPossSWM problem, i.e., checking whether the given committee is possible to be $\SWM$ is NP-complete, even for $k=1$ and when each agent's approval set is of size at most $3$. 
	
	We reduce from the classic NP-complete problem Exact Cover by 3-Sets (X3C) \citep{GaJo79a}. An X3C instance involves $3q$ elements in the ground set $U$ and a family of sets $S$ consisting of subsets of $U$ of size $3$. A subset $T$ of $S$ is an exact cover if each element of $U$ is contained in exactly one subset in $T$. The question is whether a given instance of X3C admits an exact cover.
	
	We reduce from an instance of X3C to an instance of \isPossSWM as follows.
	Let $C = U \cup \{w\}$ be the set of candidate, $k=1$, and $W=\{w\}$ be the given committee. The set of voters $V$ has $q+1$ voters. The first $q$ voters have $|S|$ different possible approval sets, each equal to one subset in $S$. Hence, none of the first $q$ voters has $w$ in any of their possible approval sets. The voter $q+1$ has only approval set, that being $W=\{w\}$. 
	
	Observe that the social welfare generated by $W$ in any of the plausible approval profiles is exactly $1$. Therefore, $W$ is a possibly $\SWM$ outcome \textit{if and only if} there is a plausible approval profile such that no candidate in $U$ has a welfare contribution of more than $1$ which is equivalent to saying that no candidate in $U$ is approved by more than $1$ voter (as otherwise, picking that candidate for the committee generates social welfare larger than $1$). 
	
	We prove that we have a YES X3C instance if and only if $W$ is a possibly $\SWM$ outcome. 
	
	($\implies$) We have a YES instance of X3C, and hence an exact cover $T$ which must contain exactly $q$ subsets of $S$. We create a plausible approval profile $\profile$ by assigning to each of the first $q$ voters a unique subset in $T$. Voter $q+1$ has approval set $\{w\}$ by construction. Note that each candidate in $C$ is approved by exactly one voter, hence any committee of size $1$ generates social welfare of $1$. Therefore $W$ is a $\SWM$ outcome in $\profile$.
	
	($\impliedby$) $W$ is a possibly $\SWM$ outcome, hence there is a plausible approval profile $\profile$ such that no candidate in $U$ is approved by more than one voter. Let $T$ contain the approval sets of the first $q$ voters. No two subsets in $T$ have an element in common and each has $3$ elements, hence each of the $3q$ elements of $S$ must appear in exactly one subset of $T$ and subsequently $T$ is an exact cover.
\end{proof}

\subsection{Proof of Corollary~\ref{coro:swm_prob_lottery_sharp_p_complete}}
\begin{proof}
	Recall the reduction in the proof of \Cref{SWM-Prob-Sharp-P-c}. There is one-to-one correspondence between plausible approval profiles on which $W$ is $\SWM$ and the X3C solutions. In the constructed instance, since the first $q$ voters each have $S$ approval sets with equal  probabilities while the $(q+1)$-th voter has only one certain approval set,   there are $|S|^q$ plausible approval profiles each  with probability $\frac{1}{|S|^q}$. 
	Consequently, the problem of computing $\Prob[W\text{ is } \SWM]$ for the given committee $W$ can be represented as 
	\begin{align*}
		\Prob[W\text{ is } \SWM]&=\sum_{\profile \in \profileset}\Delta(\profile) \cdot \mathbb{I}(\textsc{IsSWM}(W, \profile)) \\
		&=\frac{1}{|S|^q}\sum_{\profile\in \profileset}\mathbb{I}(\textsc{IsSWM}(W, \profile)).
	\end{align*}
	Thus, computing $\Prob[W\text{ is } \SWM]$ is equivalent to counting the number of plausible approval profiles in which $W$ is $\SWM$, which is further equivalent to counting the number of exact covers by {\sc X3C}. Since the problem of counting the number of exact covers, i.e., {\sc \#X3C}, is well-known to be \#P-complete, we conclude that computing $\Prob[W\text{ is } \SWM]$ under the Lottery model is \#P-complete.
\end{proof}

\subsection{Proof of Lemma~\ref{IsNecSWM_lemma}}
\begin{proof}
	$(\implies)$ Proof by contradiction. Assume that $W$ is a necessarily $\SWM$ outcome but there exists a pair of candidates $(c,c')$ where $c\in C$ and $c^\prime \in C\setminus W$ such that there is an approval profile $\profile$, in which $\AS(c^\prime,\profile) > \AS(c,\profile)$. Consider the committee $W^\prime=W\setminus \{c\} \cup \{c^\prime\}$. $W$ is not $\SWM$ in $\profile$ because $\SW(W^\prime,\profile) = \SW(W,\profile) + (\AS(c^\prime,\profile) - \AS(c,\profile)) > \SW(W,\profile)$, implying that $W$ is not a necessarily $\SWM$ outcome, a contradiction. \\
	$(\impliedby)$ Assume that a given committee $W$ satisfies that for every plausible approval profile $\profile$, $\forall~ c\in W$ and $\forall c^\prime \in C\setminus W$, we have $\AS(c,\profile) \geq \AS(c^\prime,\profile)$. Then for any other committee $W^\prime$, $\SW(W,\profile)=\sum_{c\in W}\AS(c,\profile) \geq \sum_{c^\prime \in W^\prime} \AS(c',\profile) =\SW(W^\prime,\profile)$. Therefore $W$ generates the highest social welfare in every plausible approval profile and hence is a necessarily $\SWM$ outcome.
\end{proof}

\subsection{Proof of \Cref{thm:IsNecSWM-Lot}}
\begin{proof}
	Following lemma~\ref{IsNecSWM_lemma}, to prove that \isNecSWM can be solved in polynomial time for any given committee $W$, it is sufficient to show that 
	we can verify in polynomial time whether for all candidate pairs $(c, c^\prime),c\in W,c^\prime \in C\setminus W$, and all plausible approval profiles $\profile$, it holds that $\AS(c,\profile) \geq \AS(c^\prime,\profile)$.
	To show the latter, 
	we construct a \textit{deterministic} approval profile $\bar{\profile}$ for every candidate pair $(c, c^\prime)$ and demonstrate that the pair $(c, c^\prime)$ satisfies $\AS(c,\profile) \geq \AS(c^\prime,\profile)$ for all plausible profile $\profile$ if and only if $(c, c^\prime)$ satisfies $\AS(c,\bar{\profile}) \geq \AS(c^\prime,\bar{\profile})$ for the constructed approval profile $\bar{\profile}$.
	
	\noindent \textbf{Deterministic profile construction.} Given a committee $W$, consider each pair $(c, c^\prime)$ where $c \in W$ and $c^\prime \in C \setminus W$. For each voter’s plausible approval set $A_i$, there are four possible cases: (1) $c^\prime \in A_i$ and $c \notin A_i$; (2) $c^\prime \in A_i$ and $c \in A_i$; (3) $c^\prime \notin A_i$ and $c \notin A_i$; (4) $c^\prime \notin A_i$ and $c \in A_i$. We construct the deterministic approval profile $\bar{\profile}$ as follows: for each voter $i$, set $\bar{A_i}$ by selecting a plausible approval set in the following priority order: (1) $\succ$ (2) $\succ$ (3) $\succ$ (4). That is, we first check whether there exists a plausible approval set such that $c^\prime$ is in the approval set while $c$ is not. If such an approval set exists, we set it as $\bar{A}_i$ in the deterministic approval profile $\bar{\profile}$; otherwise, we consider cases (2), (3), and (4) in sequence.	
	
	Next, we prove that if a pair $(c, c^\prime)$ satisfies $\AS(c,\bar{\profile}) \geq \AS(c^\prime,\bar{\profile})$ in the constructed approval profile $\bar{\profile}$, then the pair $(c, c^\prime)$ satisfies $\AS(c, \profile) \geq \AS(c^\prime, \profile)$ for all plausible approval profiles $\profile$.
	
	$(\implies)$ The proof is straightforward. If a pair $(c, c^\prime)$ satisfies $\AS(c,\profile) \geq \AS(c^\prime,\profile)$ for all plausible approval profiles, then the pair satisfies the condition for the plausible approval profile $\bar{\profile}$.
	
	$(\impliedby)$ Proof by contradiction. Suppose that for a pair of candidates pair $(c, c^\prime)$ we have that $\AS(c,\bar{\profile}) \geq \AS(c^\prime,\bar{\profile})$ for the constructed approval profile $\bar{\profile}$, but there exists a plausible approval profile $\profile^\prime$ such that $\AS(c,\profile') < \AS(c^\prime,\profile')$ in $\profile^\prime$, i.e., $\sum_{i\in V}\mathbb{I}[c\in A_i^\prime] < \sum_{i\in V} \mathbb{I}[c^\prime \in A_i^\prime]$ ($\mathbb{I}[\cdot]$ is an indicator function). For each $A_i^\prime$, there could only be four cases: (1) $c'\in A_i',c \notin A_i'$; (2) $c'\in A_i',c \in A_i'$; (3) $c'\notin A_i',c \notin A_i'$; and (4) $c'\notin A_i',c \in A_i'$. The inequality $\sum_{i\in V}\mathbb{I}[c\in A_i^\prime] < \sum_{i\in V} \mathbb{I}[c^\prime \in A_i^\prime]$ implies that the number of cases of type (1) must be strictly greater than the number of cases of type (4) in $\profile^\prime$ because the effects of cases (2) and (3) are canceled out in the inequality. Recall the construction of the approval profile $\bar{\profile}$ in which case (1) has the highest priority to be chosen into the deterministic approval profile $\bar{\profile}$, i.e., for each voter $i$, we preferentially set $\bar{A}_i$ as the possible approval set including $c^\prime$ while excluding $c$. Therefore, the number of instances of case (1) will still be strictly greater than the number of instances of case (4) in the approval profile $\bar{\profile}$, i.e., $\sum_{i \in V} \mathbb{I}[c \in \bar{A}_i] < \sum_{i \in V} \mathbb{I}[c' \in \bar{A}_i]$. This implies $\AS(c,\bar{\profile}) < \AS(c',\bar{\profile})$, contradicting our assumption that $\AS(c,\bar{\profile}) \geq \AS(c',\bar{\profile})$.
	
	Notice that	there are polynomial many candidate pairs $(c,c')$, and for each pair, constructing the approval profile $\bar{A}$ and verifying whether $\AS(c,\bar{\profile}) \geq \AS(c^\prime,\bar{\profile})$ can be done in polynomial time. Hence \isNecSWM is in P under the lottery model.
\end{proof}

\section{Omitted Proofs from Section 5}
\subsection{Proof of \Cref{th:MaxSWM-Lottery}}
\begin{proof}
	We reduce from the {\sc Min-$r$-Union} (M$r$U) problem \citep{VINT02a}: given a universe set $U$ of $m$ elements, a collection of $q$ sets $\mathcal{S}=\{S_1,\ldots, S_q\}$, $\mathcal{S}\subseteq 2^U, \forall \, i\in [q], S_i \subseteq U$, and two integer $r \leq q$ and $\ell$. The goal is to decide whether there exists a collection $\mathcal{I}$ with $|\mathcal{I}| = r$ such that $|\bigcup_{i\in \mathcal{I}} S_i| \leq \ell$. The mapping is as follows, we construct one sole voter with a lottery profile $\{(\frac{1}{q}, A_i)\}_{i\in [q]}$ where for each $i\in [q]$, $A_i$ is set to $S_i$. The set of candidate $C$ is set to $U$. Next, we prove that we have a yes instance of the M$r$U problem if and only if there exists a $W$ of size $\ell$ satisfying $\Prob \left[W~ \text{is}~ \SWM \right] \geq \frac{r}{q}$.
	
	($\implies$) Given an YES instance of M$r$U problem, then there exists a collection $\mathcal{I}$ with size $r$ and $|\bigcup_{i\in \mathcal{I}} S_i| \leq \ell$. We claim that there exists a committee $W$ such that $\Prob\left[W~ \text{is}~ \SWM \right] \geq \frac{r}{q}$. There are two possible cases. (1). $|\bigcup_{i\in \mathcal{I}}S_i| = \ell$, let $W=\bigcup_{i\in \mathcal{I}}S_i$. Since $W$ covers all the candidates in each profile in $\{A_i=S_i\}_{S_i\in \mathcal{I}}$, $W$ maximizes the social welfare for these realizations. (2). $|\bigcup_{i\in \mathcal{I}}S_i| \leq \ell$, let $W=\left(\bigcup_{i\in \mathcal{I}}S_i\right) \cup \bar{W}$ where $\bar{W}$ is a complement candidate set by choosing arbitrary $k - \ell$ candidates from $C\setminus \left(\bigcup_{i\in \mathcal{I}}S_i\right)$. In this case, $W$ is still social welfare maximizer for profiles in $\{A_i=S_i\}_{i\in \mathcal{I}}$. Totally, there are $q$ lottery profiles with equal probability. Therefore, the probability $\Prob\left[W~ \text{is}~ \SWM \right] \geq \frac{r}{q}$.
	
	($\impliedby$) For the \maxSWM problem, if there exists a $W$ such that $\Prob\left[W~ \text{is}~ \SWM \right]  \geq \frac{r}{q}$, as each approval profile $A_i$ is equally probable, $W$ maximizes social welfare for at least $r$ realized profiles over all $q$ possible profiles. W.l.o.g, choose $r$ approval profiles $ \mathcal{P}=\{A_1,\ldots, A_r\}$ in which $W$ maximizes the social welfare. Since for each $A_i$, $|A_i| \leq k$, then $W$ maximizing social welfare for each $A_i$ implies $W$ contains all the candidates in each approval set $A_i$, i.e., $A_i\subseteq W$. This implies $\bigcup_{A_i\in \mathcal{P}}A_i \subseteq W$. 
	Notably, from $\mathcal{A}$ to $\mathcal{S}$, the mapping from $A_i$ to $S_i$ is one-to-one. Therefore, we find a collection $\mathcal{I}$ such that $|\bigcup_{i \in \mathcal{I}} S_i| \leq k = \ell$.
\end{proof}

\subsection{Proof of Lemma~\ref{lemma:exist_nec_swm_no_arc}}
\begin{proof}
	Consider any arbitrary candidate set $\bar{W} \subseteq C\setminus\{c_i,c_j\}$, $|\bar{W}|=k-1$. Denote $W^i=\bar{W}\cup\{c_i\}$ as any committee including $c_i$ but excluding $c_j$; $W^j=\bar{W}\cup\{c_j\}$ as any committee including $c_j$ but excluding $c_i$. No edge between $c_i$ and $c_j$ implies that there exists some deterministic approval profile $\profile^i$ such that $\AS(c_i,\profile^i) > \AS(c_j,\profile^i)$ while there exists some deterministic approval profile $\profile^j$ such that $\AS(c_j,\profile^j) > \AS(c_i,\profile^j)$. We show that neither $W^i$ nor $W^j$ can be a necessarily $\SWM$ committee. $W^i$ is not $\SWM$ in approval profile $\profile^j$ as $\SW(W^i,\profile^j)=\sum_{c\in \bar{W}}\AS(c,\profile^j) + \AS(c_i,\profile^j) < \sum_{c\in \bar{W}}\AS(c,\profile^j) + \AS(c_j,\profile^j)=\SW(W^j,\profile^j)$ while $W^j$ does not maximize the social welfare in approval profile $\profile^i$ as $\SW(W^j,\profile^i)=\sum_{c\in \bar{W}}\AS(c,\profile^i) + \AS(c_j,\profile^i)$ which is smaller than $\sum_{c\in \bar{W}}\AS(c,\profile^j) + \AS(c_i,\profile^i)=\SW(W^i,\profile^i)$. Note that $\bar{W}$ is chosen arbitrarily. Therefore, 
	the fact that neither $W^i$ nor $W^j$ is a necessarily $\SWM$ committee implies that 
	any committee with $c_i$ but without $c_j$ or with $c_j$ but without $c_i$ can never be a necessarily $\SWM$ committee.
\end{proof}

\subsection{Proof of \Cref{thm:ExistsNecSWM-Lot}}
\begin{proof}
	Based on \Cref{algo:ExistsNecSWM-Lot}, we show that under the lottery model, \existsNecSWM returns YES if and only if \Cref{algo:ExistsNecSWM-Lot} returns YES.
	
	$(\implies)$ We prove by showing that if \Cref{algo:ExistsNecSWM-Lot} returns NO, there is no necessarily $\SWM$ committee. If \Cref{algo:ExistsNecSWM-Lot} computes a committee $W$ and returns NO, then there exists some $c\in W, c^\prime \in C\setminus W$, $(c,c^\prime)\notin E$. Since \Cref{algo:ExistsNecSWM-Lot} selects zero indegree candidates in $G$ iteratively and $c^\prime \in C\setminus W$, $(c^\prime ,c) \notin E$ (otherwise the indegree of $c$ is not $0$). So, $(c,c^\prime)$ constructs a pair of candidates without a domination relation. Now we assume there exists a necessarily $\SWM$ committee $W^\ast$, according to lemma~\ref{lemma:exist_nec_swm_no_arc}, there are two cases for $c$ and $c'$. \\
	\textbf{Case 1}: $c\in W^\ast, c^\prime \in W^\ast$. In this case, there must exist some candidate $c'' \in W$ and $c'' \notin W^\ast$. Since $W^\ast$ is a necessarily $\SWM$ committee, we have $\AS(c^\prime,\profile) \geq \AS(c'',\profile)$ for every plausible approval profile $\profile$ (via lemma~\ref{IsNecSWM_lemma}). On the other hand, according to \Cref{algo:ExistsNecSWM-Lot}, it must be that $(c',c'')\notin E$ (as $c'\notin W$, if $(c',c'')\in E$, the indegree of $c''$ is strictly larger than $0$, and then can not be selected in $W$). Thus, the only possible case is that $c'$ and $c''$ has the same approval score for every plausible approval profile and $c''$ has higher lexicographic order than $c'$ according to the tie-breaking rule. Then we have $(c'',c^\prime) \in E$. Notice that there is no edge between $c$ and $c'$, which means no ties between $c$ and $c'$. So, there is no ties between $c$ and $c''$. Recall that  $W^\ast$ is a necessarily $\SWM$ committee. We have $\AS(c, \profile) \geq \AS(c'',\profile)$ for every plausible approval profile $\profile$ and there is no ties between $c$ and $c''$. Then, $(c,c'')\in E$. However, $(c,c'')\in E$ and $(c'',c^\prime) \in E$ imply that $(c,c^\prime)\in E$ by transitivity, contradicting the condition that there is no edge between $c$ and $c^\prime$.\\
	\textbf{Case 2}: $c \notin W^\ast, c^\prime \notin W^\ast$. Denote $S = W^\ast \setminus W$. According to \Cref{algo:ExistsNecSWM-Lot}, for every $c''\in S$, $(c'',c)\notin E$ because $c$ is selected in $W$ while $c''$ is not. However, as we assume that $W^\ast$ is a necessarily $\SWM$ committee, by lemma~\ref{IsNecSWM_lemma} it must hold that $\AS(c'', \profile) \geq \AS(c, \profile)$ and $\AS(c'',\profile) \geq \AS(c',\profile)$ for every plausible approval profile. With regard to $c$ and $c''$, $(c'',c)\notin E$ implies that the only feasible case is that $\AS(c'', \profile) = \AS(c, \profile)$ for every plausible approval profile and that $c$ has higher priority in the tie-breaking. Then, we have $(c,c'')\in E$. Notice that there is no edge between $c$ and $c'$, implying that there is no tie between $c$ and $c'$. It also implies that there will be no tie-breaking between $c'$ and $c''$, i.e., $(c'',c')\in E$ because $\AS(c'',\profile) \geq \AS(c',\profile)$ for every plausible approval profile. Then we can deduce $(c,c')\in E$ from $(c,c'')\in E$ and $(c'',c')\in E$ by transitivity. This contradicts the assumption that there is no edge between $c$ and $c'$.
	
	$(\impliedby)$ If \Cref{algo:ExistsNecSWM-Lot} returns YES, the committee $W, |W|=k$, is a necessarily $\SWM$ committee according to lemma~\ref{IsNecSWM_lemma}. Therefore, \existsNecSWM returns YES. 
\end{proof}

\subsection{Proof of Proposition~\ref{proposition_candidate_probability_maxSWM_1}}\label{sec:maxswm_cp_n_1}
\begin{proposition}\label{proposition_candidate_probability_maxSWM_1}
	Under the Candidate Probability model, \maxSWM is solvable in polynomial time, when $n=1$. 
\end{proposition}
\begin{proof}
	We show \maxSWM is solvable in polynomial time when $n=1$ by establishing the statement: when $n=1$, the committee $W$ that maximizes the probability of being $\SWM$ corresponds to the top-$k$ candidates with the highest approval probabilities. Let $A_1$ be the approval set obtained by sampling from the candidate probability model.  
	The key observation is that, if a committee $W$ satisfies $\SWM$, then either (1) $W$ is the subset of the approval set $A_1$ of the unique voter $1$ (when the size of the realization of the approval set is larger than $k$) or (2) $A_1$ is a subset of the committee $W$ (when the size of the approval set is smaller than $k$). 
	It follows that the probability of $W$ being $\SWM$ can be represented as follows.
	
	\begin{align}
		\Prob\left[W \text{ is }\SWM \right] 
		&=\Prob\left[W \subset A_1 \text{ and } |A_1| > k \right] + \Prob\left[A_1 \subseteq W \text{ and } |A_1| \leq k \right]\notag \notag \\
		&= \Prob\left[W \subset A_1 \mid |A_1| > k \right] \cdot \Prob\left[|A| > k\right] \notag \\
		& \quad + \Prob\left[A_1 \subseteq W \mid |A_1| \leq k\right] \cdot \Prob\left[|A| \leq k\right] \notag \\
		&= \Prob\left[W \subset A_1\right]\cdot \Prob\left[|A_1| > k \right]+ \Prob\left[A_1 \subseteq W\right]\cdot \Prob\left[|A_1| \leq k\right] \notag \\
		&=   \Big(\prod_{c\in W} p_{1,c} \Big) \cdot \Prob\left[|A_1| > k\right] +   \Big(\prod_{c^\prime \in C \setminus W} (1 - p_{1,c^\prime})\Big) \cdot \Prob\left[|A_1| \leq k\right] . \label{eq:SWMn=1}
	\end{align}
	We claim that \Cref{eq:SWMn=1} is maximized by selecting the top $k$ candidates with the highest approval probabilities. Let $W$ denote the set of these top-$k$ candidates. Suppose, for the sake of contradiction, that there exists another set $W' \neq W$ such that $W'$ maximizes the probability of being a $\SWM$ committee.
	
	Consider any arbitrary candidate $c_1 \in W \setminus W'$ and any arbitrary candidate $c_2 \in W' \setminus W$ such that $p_{1,c_1} \geq p_{1,c_2} \geq 0$. We have $\prod_{c\in W'}p_{1,c}$ and $\prod_{c'\in C\setminus W'} (1-p_{1,c'})$ are non-decreasing by replacing $c_2$ with $c_1$ in $W'$ as $p_{1,c_1} \geq p_{1,c_2} \geq 0$. This contradicts the assumption that $W'$ maximizes the probability of being $\SWM$. Therefore, we conclude that when $n=1$, the committee consisting of the top-$k$ highest approval probabilities ($p_{i,c}$) candidates maximizes the probability of being $\SWM$. This selection can be efficiently implemented via sorting in polynomial time.
\end{proof}

\subsection{Proof of Proposition~\ref{th:MaxSWM-CP_constant_k}}\label{sec:maxswm_cp_constant_k}
\begin{proposition}\label{th:MaxSWM-CP_constant_k}
	Under the Candidate Probability model, for constant $k$, \maxSWM is solvable in polynomial time.
\end{proposition}
\begin{proof}
	When the committee size $k$ is constant, there are at most $O(m^k)$ possible committees. For each committee, we can use binary search to compute its probability of being $\SWM$ in polynomial time via \Cref{thm:MaxSW-Prob-CP-3VA}. This implies that \maxSWM is tractable with constant $k$ since we can compute the probabilities of all $O(m^k)$ committees being $\SWM$.
\end{proof}

\subsection{Proof of \Cref{thm:ExistsNecSWM-CP-3VA}}\label{sec:proof_existsNecSWM_CP}
\begin{proof}
	We first describe a polynomial-time algorithm to decide the \existsNecSWM problem under the Candidate Probability model. The algorithm procedures are as follows.
	\begin{itemize}
		\item Construct a \emph{deterministic} profile $\bar{\profile}=(\bar{A}_1,\cdots, \bar{A}_n)$ with only certain approval ballots, that is, for every voter $i$, $\bar{A}_i=\{c\in C:  p_{i,c}=1\}$.
		\item Compute the $\SWM$ committee $W^\ast$ in $\bar{\profile}$, breaking ties by selecting the committee with the greatest number of approvals with positive probabilities, i.e., $$W^\ast= \underset{W \text{ is } \SWM}{\arg\max}~ \sum_{i\in V}|\{c\in W : p_{i,c} > 0\}|.$$
		\item Return YES if $W^\ast$ is certificate of YES instance for \isNecSWM, otherwise NO.
	\end{itemize}
	
	Next, we prove the following statement: Under the candidate probability model, \existsNecSWM returns YES if and only if $W^\ast$ is a necessarily $\SWM$ committee. 
	
	$(\implies)$ Prove by contrapositive: if $W^\ast$ is not necessarily $\SWM$, then \existsNecSWM always returns NO. 
	Suppose for contradiction that $W^\ast$ fails to satisfy necessarily $\SWM$ but \existsNecSWM returns YES. It means there exists $W^\prime \neq W^\ast$ which is a necessarily $\SWM$ committee. Since $W^\prime$ is necessarily $\SWM$ and $W^\ast$ is $\SWM$ in $\bar{\profile}$, then $W^\prime$ must satisfy $\SW(W^\prime,\bar{\profile}) =\SW(W^\ast,\bar{\profile})$. Additionally, in the construction of $\bar{\profile}$, it only considers certain approval ballots, which implies $W^\ast$ and $W^\prime$ have the same number of certain approvals ($p_{i,c}=1$). Recall the computation of $W^\ast$, we break ties by choosing the committee with the greatest number of approvals with positive probabilities. Now consider the uncertain approvals ($p_{i,c}\in (0,1)$) in $W^\ast$ and $W^\prime$. Let $T$ be the uncertain approvals of $W^\prime$: $T = \{(i,c) : i\in V,c\in W^\prime, p_{i,c}\in (0,1)\}$ and $S$ be the uncertain approvals of $W^\ast$, i.e., $S=\{(i,c) : i\in V,c\in W^\ast, p_{i,c}\in (0,1)\}$. Then $|S|\geq |T|$. Now focus on another \textit{deterministic} approval profile $\profile^\prime$ where for each voter $i$, $A_i^\prime=\{c : c\in C, p_{i,c}=1 \text{ or } c\in W^\ast, p_{i,c} > 0 \}$, if $|S| \geq |T| > 0$, then $\SW(W^\ast, \profile^\prime) > \SW(W^\prime, \profile^\prime)$ because the realization of profile $\profile^\prime$ only converts uncertain approvals w.r.t. $W^\ast$ into certain approvals. This contradicts to $W^\prime$ is a necessarily $\SWM$ committee. The other case is $|S| = |T|= 0$, meaning $W^\prime$ and $W^\ast$ only have certain approvals. This implies for any plausible approval profile $\profile$, $\SW(W^\prime, \profile)=\SW(W^\ast,\profile)$, contradicting to the assumption that $W^\prime$ is necessarily $\SWM$ while $W^\ast$ is not.
	
	$(\impliedby)$ If $W^\ast$ is necessarily $\SWM$, then there indeed exists a committee which is necessarily $\SWM$. Therefore \existsNecSWM returns YES.
\end{proof}

\end{document}